\documentclass[a4paper,10pt]{article}
\usepackage[T1]{fontenc}
\usepackage{etoolbox}
\usepackage{amsmath,amsthm,amssymb,mathtools,nicefrac}
\usepackage{stmaryrd,enumitem}
\usepackage{scalerel,cmll,lmodern}
\usepackage{graphicx,epsfig,tensor}
\usepackage{multicol,accents,float,threeparttable}
\usepackage{url}
\allowdisplaybreaks

\usepackage{geometry}\usepackage{footnote}
\makesavenoteenv{tabular}

\usepackage[ruled]{algorithm2e}
\usepackage{adjustbox}

\DeclareFontFamily{U}{pxsyb}{} \DeclareFontShape{U}{pxsyb}{m}{n}%
{<-> s * [1] pxsyb}{}
\DeclareMathAlphabet{\me}{U}{pxsyb}{m}{n}

\DeclareSymbolFont{OilerScript}{U}{eus}{m}{n}
\SetSymbolFont{OilerScript}{bold}{U}{eus}{b}{n}
\DeclareSymbolFontAlphabet\mathgra{OilerScript}

\newcommand{\mN}{\mathbb{N}}

\newcommand{\mR}{\mathbb{R}}
\newcommand{\mS}{\mathbb{S}}
\newcommand{\mT}{\mathbb{T}}

\newcommand{\fC}{\fset{C}}

\newcommand{\kE}{\mathcal{E}}

\newcommand{\Ff}{\mathsf{F}}

\newcommand{\Of}{\mathsf{O}}

\newcommand{\ep}{\epsilon}

\newcommand{\thet}{\theta}

\newcommand{\ro}{\varrho}

\newcommand{\la}{\lambda}

\newcommand{\fii}{\phi}

\newcommand{\Fi}{\Phi}
\newcommand{\Ga}{\Gamma}

\newcommand{\Sig}{\Sigma}
\newcommand{\Yp}{\Upsilon}

\newcommand{\ac}{\acute}

\newcommand{\dhat}{\widehat}
\newcommand{\dtl}{\widetilde}

\makeatletter
\newcommand{\raisemath}[1]{\mathpalette{\raisemth{#1}}}
\newcommand{\raisemth}[3]{\raisebox{#1}{$#2#3$}}
\makeatother

\newcommand{\gd}{\!\raisemath{-3.5pt}{\scaleobj{2.4}{\cdot}}\!}

\newcommand{\norm}[1]{\|#1\|}

\newcommand{\dt}{\cdot}

\newcommand{\mi}{\wedge}
\newcommand{\ma}{\vee}

\newcommand{\fset}[1]{{\me{#1}}}

\DeclareMathAlphabet\mathbf{OT1}{cmr}{bx}{n}

\newcommand{\limsp}{\operatornamewithlimits{\mathop{limsup}}}

\DeclareMathOperator{\Sy}{\bm{\mathsf{Sy}}}

\DeclareMathOperator{\rank}{rank}

\newcommand{\Row}{\operatornamewithlimits{\mathop{\bm{\mathsf{Row}}}}}

\newcommand{\Col}{\operatornamewithlimits{\mathop{\bm{\mathsf{Col}}}}}

\newcommand{\rlm}{\mathord{\downarrow}}

\newcommand{\dif}{\mathsf{d}}

\makeatletter

\usepackage{bm,natbib}
\usepackage[hidelinks]{hyperref}

\newtheorem{thm}{Theorem}
\newtheorem{lem}{Lemma}

\theoremstyle{remark}

\theoremstyle{plain}

\theoremstyle{remark}

\theoremstyle{definition}
\newtheorem{defi}{Definition}

\theoremstyle{definition}
\newtheorem{as}{Assumption}

\theoremstyle{plain}

\newtheorem{coro}{Corollary}

\begin{document}

\title{\textbf{Dissipative Stability Conditions for Linear Coupled Differential-Difference Systems via a Dynamical Constraints Approach}}

\author{Qian Feng\thanks{qfen204@aucklanduni.ac.nz, Department of Electrical and Computer Engineering, The University of Auckland, Auckland 1010,  New Zealand}}
\date{}

\maketitle

\begin{abstract}
\noindent In this short note, we derive dissipative conditions with slack variables for a linear coupled differential-difference (CDDS) via constructing a Krasovskii functional. The approach can be interpreted as a generalization of the Finsler Lemma approach for standard LTI systems proposed previously in \cite{de2001stability}.  We also show that the proposed slack variables scheme is equivalent to the approach based on directly substituting the system trajectory $\dot{\bm{x}}(t)$, similar to the case of LTI system.


\end{abstract}

\section{Introduction}
Many approaches for the stability analysis of time delay system \cite{fridman2014introduction} have been proposed over the recent decades based on solving semidefinite programs derived via constructing Krasovskii functional. Most of such results \cite{Briat2014} are obtained by directly substituting $\dot{\bm{x}}(t)$ with the system trajectory expression during the process of constructions. In this note however, we will exploit the idea of Finsler Lemma (Projection Lemma) approach, originated from \cite{de2001stability} concerning LTI systems, to derive dissipative conditions for a linear coupled differential-difference system \cite{Gu2009798} with a distributed delay.  The dissipative conditions derived via Finsler Lemma in this note, which is also based on constructing a Krasovskii functional,  can be considered as a generalization of the exiting results for LTI systems in \cite{de2001stability}. The application of Finsler Lemma with Projection Lemma avoids a direct substitution of $\dot{\bm{x}}(t)$, which results in the introduction of extra-matrix terms in the resulting dissipative conditions. Finally, we show that the dissipative conditions with slack variables are equivalent to the conditions derived via the approach of directly substituting $\dot{\bm{x}}(t)$. For the advantage of having slack variables in dissipative conditions, see the examples in explained in \cite{de2001stability} which can also be valid for the corresponding delay problems. 

\section*{Notation}
We define \( \mathbb{T}: =\{ x \in \mathbb{R}: x \geq 0 \}  \) and  \( \mathbb{S}^{n}:= \{ X \in \mathbb{R}^{n \times n}: X = X^\top \}  \) and \( \mathbb{R}_{[n]}^{n \times n}: = \{ X \in \mathbb{R}^{n \times n}: \rank(X) = n\}  \).  The notations \( \norm{\mathbf{x}}_{q} = \left( \sum_{i=1}^{n}|x_i|^{q} \right)^\frac{1}{q} \) and \( \norm{f(\cdot)}_{p} = \left( \int_{ \mathbb{R}} |f(t)|^p \dif t \right)^{\frac{1}{p}} \) and 
\( \norm{\bm{f}(\cdot)}_{p} = \left( \int_{ \mathbb{R}} \|\bm{f}(t)\|_{2}^p \dif t \right)^{\frac{1}{p}} \) are the norms associated with \( \mathbb{R}^{n}\) and 
Lebesgue integrable functions space $ \fset{L}_{p}(\mathbb{R} ; \mathbb{R})$ and $ \fset{L}_{p}(\mathbb{R} ; \mathbb{R}^n)$, respectively. In addition, We use \( \fset{C}_{\circ}^{\bullet}(\mathcal{X};\mathbb{R}^n) \) to denote the space of right piecewise continuous functions. Let \( \Sy(X) := X + X^\top \) to be the sum of a matrix with its transpose. \( \Col_{i=1}^{n} x_i := \left[ \Row_{i=1}^{n} x^\top_i \right]^\top = \big[ x_1^\top \cdots  x_i^\top  \cdots  x_n^\top \big]^\top \) is defined a column vector containing a sequence of objects. The symbol \( * \) is applied to denote  \( [*] Y X = X^\top YX  \) or \( X^\top Y[*] = X^\top YX \). Furthermore, let $x \ma y = \max(x,y)$ and $x \mi y  = \min(x,y)$.  $\Of_{n \times n}$ is used to denote a $n\times n$ zero matrix with the abbreviation $\Of_n$, whereas $\bm{0}_n$ denotes a $n \times 1$ column vector. The diagonal sum of two matrices and \( n \) matrices are defined as \( X \oplus Y = \mathsf{Diag} (X, Y),\; \bigoplus_{i=1}^n X_i = \mathsf{Diag}_{i=1}^{n} (X_i) \), respectively. Furthermore, \( \otimes \) stands for the Kronecker product. Finally, we assume the order of matrix operations as \emph{matrix (scalars) multiplications} $> \otimes > \oplus > $ \emph{matrix (scalar) additions}.

\section{Problem formulation} 

Consider a linear coupled differential-difference system (CDDS) 
\begin{align}
& \dot{\bm{x}}(t) = A_1\bm{x}(t) +  A_2 \bm{y}(t-r) + \textstyle \int_{-r}^{0} \ac{A}_3 (\tau) \bm{y}(t+\tau) \dif \tau +  D_1 \bm{w}(t), \notag
\\
& \bm{y}(t) = A_4\bm{x}(t) + A_5\bm{y}(t-r), \label{openloop}
\\
& \bm{z}(t) = C_1\bm{x}(t) + C_2 \bm{y}(t-r) + \textstyle \int_{-r}^{0} \ac{C}_3 (\tau) \bm{y}(t+\tau) \dif \tau + D_2 \bm{w}(t), \notag
\\[1mm]
& \left( \bm{x}(0), \bm{y}(0+\dt) \right) = \left( \bm{\xi}, \bm{\phi}(\dt) \right) \in \mathbb{R}^{n} \times \fC_{\circ}^{\bullet}([-r,0) \fatsemi \mathbb{R}^{\nu}), \notag
\end{align} 
where \( \bm{x}(t) \in \mathbb{R}^{n} \) and $\bm{y}(t) \in \mathbb{R}^{\nu}$ are the solution of \eqref{openloop},
 \( \bm{w}(\cdot) \in \fset{L}_{\mathit{2}} (\mathbb{T} \fatsemi \mathbb{R}^q) \) represents disturbance, 
\( \bm{z}(t) \in \mathbb{R}^m \) is the regulated output. Furthermore, $\bm{\xi} \in \mR^{n}$ and  
\( \bm{\phi}(\cdot) \in \fC_{\circ}^{\bullet}([-r,0) \fatsemi \mathbb{R}^n) \) are the initial conditions with a known delay value $r>0$. The distributed delay term is $F(\tau):=\bm{f}(\tau) \otimes I_{\nu}$ with $\bm{f}(\tau) = \Col_{i=1}^{d}f_i(\tau) \in \fC^{1}([-r,0] \fatsemi \mathbb{R}^d)$. Moreover, the state space matrices are 
$A_1 \in \mathbb{R}^{n \times n}$ with $ A_2; A_3(\tau); A_4 \in \mathbb{R}^{n \times \nu}$, and $C_1 \in \mathbb{R}^{m \times n}$ with $ C_2; \ac{C}_3(\tau) \in \mathbb{R}^{m \times \nu}$, and  $D_1 \in \mathbb{R}^{n \times q}$ and  $D_2 \in \mathbb{R}^{m\times q} $  with the indexes $n;\nu \in \mN$ and $m;q\in \mathbb{N}_{0} $.  It is also assumed $\norm{A_5} <1$ which ensures the input to state stability of $\bm{y}(t) = A_4\bm{x}(t) + A_5\bm{y}(t-r)$ \cite{Gu2009798}. Finally, $A_3(\tau)$ and $C_3(\tau)$ satisfy the following assumption. 

\begin{as}\label{as1}
 There exist $ \Col_{i=1}^{d} f_i(\tau) = \bm{f}(\cdot) \in \fset{C}^{1} (\mR \fatsemi \mathbb{R}^d)$ with 
$d \in \mathbb{N}_{0} $, and 
\( A_3 \in \mathbb{R}^{n \times \ro}\),  \( C_3 \in \mathbb{R}^{m \times \ro} \) with $\ro  = \nu d$ such that 
$\forall \tau \in [-r,0]$, \( \mathbb{R}^{n \times \nu} \ni \ac{A}_3(\tau) = A_3F(\tau)$ and 
\( \mathbb{R}^{m \times \nu} \ni \ac{C}_3(\tau) = C_3F(\tau) \). 
	In addition, $\{ f_i(\dt) \}_{i=1}^{d} $ are linearly independent and \( \bm{f}(\dt) \) satisfies \begin{equation}\label{pera}
\exists ! M \in \mR^{d \times d}: \dfrac{\dif \bm{f}(\tau)}{\dif \tau} = M \bm{f}(\tau). 
\end{equation}  
\end{as}

\subsection{Preliminaries}

To prove the main results in this note, the following Lemmas and Definitions are required.
\begin{lem} \label{Lemma 1}
    For all \( P \in \mathbb{R}^{p \times q} \) and \( Q \in \mathbb{R}^{n \times m} \), we have
    \begin{gather}
    (P \otimes I_n)( I_q \otimes Q )  = P \otimes Q = \left( I_p \otimes Q\right)(P \otimes I_m).  \label{commutative}
    \end{gather}
    Moreover, we have $\forall X \in \mathbb{R}^{n \times m}$, \;$\forall Y \in \mathbb{R}^{m \times p}$, \; $\forall Z \in \mR^{q\times r}$,
    \begin{gather}
    (XY) \otimes Z =  (XY) \otimes (Z I_r) =( X \otimes Z ) ( Y \otimes I_r ). \label{function}
    \end{gather}
\end{lem}
\begin{proof}
\eqref{commutative} and \eqref{function} are derived via the property of the Kronecker product: $(A \otimes B)(C \otimes D) = (AC)\otimes(BD)$. 
\end{proof}

The following Lemma \ref{Lemma 2} is a particular case of the Theorem 3 in \cite{Gu2009798}.
\begin{lem} \label{Lemma 2}
Given $r >0$, suppose a differential-difference system 
\begin{equation} \label{DDS}
\dot{\bm{x}}(t) = \bm{f}(\bm{x}(t),\bm{y}(t+\cdot)), \quad \bm{y}(t) = \bm{g}(\bm{x}(t),\bm{y}(t+\dt)), \;\; \bm{f}(\bm{0}_{n},\bm{0}_{\nu}(\dt)) = \bm{0}_n, \;\; \bm{g}(\bm{0}_{n}, \bm{0}_{\nu}(\dt)) = \bm{0}_{\nu}(\dt)    
\end{equation}
 satisfying the prerequisites in the Theorem 3 of \cite{Gu2009798},  where $\bm{y}(t+\dt) \in \fC_{\circ}^{\bullet}([-r,0) \fatsemi \mathbb{R}^{\nu})$ and $\bm{y}(t) = \bm{g}(\bm{x}(t),\bm{y}(t+\dt))$ is uniformly
    input to state stable. Then the origin of \eqref{DDS} is globally uniformly asymptotically stable, if there exists a differentiable functional $ v(\gd_1, \gd_2(\cdot) ) : \mathbb{R}^{n} \times \fC_{\circ}^{\bullet}([-r,0) \fatsemi \mathbb{R}^{\nu}) \to \mathbb{T}$ with $v(\bm{0}_{n}, \bm{0}_{\nu}(\dt)) = 0$, such that the following conditions are satisfied:
    \begin{gather}
    \begin{multlined} 
 \exists \epsilon; \ep_2 >0, \; \forall \bm{\xi} \in \mR^{n},  \; \forall \bm{\phi}(\cdot) \in \fC_{\circ}^{\bullet}([-r,0) \fatsemi \mathbb{R}^{\nu}),\;
    \epsilon_1 \norm{\bm{\xi}}_{2}^2 \leq v(\bm{\xi}, \bm{\phi}(\cdot))
    \leq \ep_2 \left(  \norm{\bm{\xi}}_2 \ma \norm{\bm{\phi}(\dt)}_{\mathit{\infty}} \right)^2 
    \end{multlined} 
    \label{posit} 
    \\[1mm]
    \;\; \exists \epsilon_3 > 0,\; \forall \bm{\xi} \in \mR^{n},
    \; \; \forall \bm{\phi}(\cdot) \in \fC_{\circ}^{\bullet}([-r,0) \fatsemi \mathbb{R}^{\nu}),\; \dot{v}( \bm{\xi}, \bm{\phi}(\cdot)) \leq - \epsilon_3 \norm{\bm{\xi}}_{2}^2 \label{neg}
    \end{gather}
    where 
\begin{equation}\label{difer}
\dot{v}(\bm{\xi}, \bm{\phi}(\cdot)):= \left.\frac{\dif^{\rlm}}{\dif t} v(\bm{x}(t),\bm{y}(t+\dt))\right|_{t = \tau, \bm{x}(\tau) = \bm{\xi}, \bm{y}(\tau+\dt) = \bm{\fii}(\dt)}, \quad \frac{\dif^{\rlm}}{\dif x} f(x) = \limsp_{\eta\rlm 0} \dfrac{f(x+\eta) - f(x)}{\eta},
\end{equation}
\end{lem}
with $\dot{\bm{x}}(t)$ and $\bm{y}(t+\dt)$ satisfying \eqref{DDS}. 

\begin{defi}[Dissipativity]\label{def1}
Given $r >0$, a delay system 
\begin{equation} \label{dissi}
\dot{\bm{x}}(t) = \bm{f}(\bm{x}(t), \bm{y}(t+\cdot), \bm{w}(t)), \quad \bm{y}(t) = \bm{g}(\bm{x}(t), \bm{y}(t+\dt)), \quad \bm{z}(t) = \bm{h}\big(\bm{x}(t), \bm{y}(t+\cdot), \bm{w}(t)\big),
 \end{equation} is dissipative with respect to the supply rate function \( s(\bm{z}(t),\bm{w}(t)) \),  if there exists a differentiable functional \( v(\gd_1, \gd_2(\cdot) ): \mathbb{R}^{n} \times \fC_{\circ}^{\bullet}([-r,0) \fatsemi \mathbb{R}^{\nu}) \to \mathbb{R} \) such that 
    \begin{equation}\label{diss}
    \forall t \in \mathbb{T}, \; \dot{v}( \bm{x}(t), \bm{y}(t+\cdot)) - s(\bm{z}(t),\bm{w}(t)) \leq 0,
    \end{equation}
with $\dot{\bm{x}}(t)$, $\bm{y}(t+\dt)$ and $\bm{z}(t)$ satisfying \eqref{dissi}. Under the assumption that \( v(\gd_1, \gd_2(\cdot) ) : \mathbb{R}^{n} \times \fC_{\circ}^{\bullet}([-r,0) \fatsemi \mathbb{R}^{\nu}) \to \mathbb{R} \) is differentiable, \eqref{diss} is equivalent to the original definition of dissipativity. (See \cite{Briat2014} for the definition of dissipativity without delays)

\end{defi}  
To conduct dissipative analysis for \eqref{openloop}, a quadratic supply function
\begin{equation}\label{supply}
s(\bm{z}(t),\bm{w}(t))  = 
\begin{bmatrix}
\bm{z}(t) \\ \bm{w}(t)
\end{bmatrix}^\top
\mathbf{J}
\begin{bmatrix}
\bm{z}(t)
\\
\bm{w}(t)
\end{bmatrix} \quad \text{with} \quad 
\mathbf{J} = \begin{bmatrix} 
J_1 & J_2 \\ * & J_3
\end{bmatrix} \in \mathbb{S}^{m + q},\; \; J_1 \preceq 0
\end{equation}
is considered in this note, which is taken from \cite{scherer1997multiobjective}.

To analyze the stability of the origin of \eqref{openloop}, we apply the Krasovskii functional
\begin{equation} \label{KL}
 v(\bm{x}(t), \bm{y}(t+\cdot)) := \begin{bmatrix}
*
\end{bmatrix} P
\begin{bmatrix}
\bm{x}(t)
\\
\int_{-r}^{0} F(\tau)\bm{y}(t+\tau) \dif \tau
\end{bmatrix} + \int_{-r}^{0}\bm{y}^\top(t+\tau)  
 \Big[ S + (\tau+r)U \Big]  \bm{y}(t+\tau)  \dif \tau
\end{equation}
to be constructed, where $P \in \mathbb{S}^{n + \ro}$ and $S;U \in \mathbb{S}^{\nu}$ and $F(\tau):= \bm{f}(\tau) \otimes I_{\nu}$ with $\bm{f}(\tau)$ which is given and defined in \eqref{openloop} satisfying Assumption \ref{as1}.
\section{Main results}

\subsection{Dissipative stability conditions without slack variables}
In this subsection, we first present dissipative conditions constructed by directly substituting the expression of $\dot{\bm{x}}(t)$ during our derivation.  

\begin{thm}\label{Theorem 1}
Given \( J_1 \prec 0\) in \eqref{supply}, the linear CDDS \eqref{openloop} is globally uniformly asymptotically stable at its origin and dissipative with respect to \eqref{supply},  if there exist  $ P \in \mathbb{S}^{n + \ro}$ and $S;U \in \mathbb{S}^{\nu}$ such that the following conditions hold,
\begin{gather}\label{positiv}
P + \Big[ \Of_n  \oplus  \left( \Ff  \otimes S \right) \Big] \succ 0,\quad S \succ 0, \;\; U \succ 0
\\
\begin{bmatrix} J_1^{-1} & \Sig \\ * & \bm{\Phi}   \end{bmatrix} \prec 0, \label{meitrix}
\end{gather}
where $\Ga := \begin{bmatrix} \Of_{\nu \times  q} &  A_4 & A_5 & \Of_{\nu \times \ro } \end{bmatrix}$ and $\Sigma := \begin{bmatrix}
D_2 & C_1 &  C_2 & C_3	
\end{bmatrix}$
and	
\begin{equation}\label{fi}
\bm{\Fi} :=  \Sy \left( H P
\Theta \right) +  \Ga^\top  \left( S + rU \right) \Ga - \Big( J_3 \oplus \Of_{n} \oplus S 
\oplus \left( \Ff  \otimes U \right)   \Big) - \Sy\left( \begin{bmatrix}  \Sig^\top J_2 & \Of_{(n + \nu + \ro + q) \times (n + \nu + \ro) } \end{bmatrix} \right).
\end{equation}
with 
\begin{equation}\label{F}
H = \left[ \begin{smallmatrix}
\Of_{ q \times n} & \Of_{q \times \ro} \\ I_n & \Of_{n \times \ro}  \\  \Of_{\nu \times n} & \Of_{\nu \times \ro} \\ \Of_{\ro \times n} & I_{\ro} \end{smallmatrix} \right], \;\; \Theta =  \begin{bmatrix}
D_1 & A_1 &  A_2 & A_3
\\[1mm]
\Of_{\ro \times q} & F(0)A_4 & F(0)A_5 - F(-r) & -\dhat{M},   
\end{bmatrix}, \;\; \Ff^{-1} = \int_{-r}^{0} \bm{f}(\tau) \bm{f}^\top(\tau) \dif \tau 
\end{equation}
\end{thm}

\begin{proof}
It is obvious to see that \eqref{KL} satisfies the property $\exists \la;\eta>0$: $\forall t \in \mT$,
\begin{multline}\label{upp}
 v(\bm{x}(t), \bm{y}(t+\dt)) \leq 
\begin{bmatrix}
\bm{x}(t)
\\
\int_{-r}^{0} F(\tau)\bm{y}(t+\tau) \dif \tau
\end{bmatrix}^\top \la \begin{bmatrix}
\bm{x}(t)
\\
\int_{-r}^{0} F(\tau)\bm{y}(t+\tau) \dif \tau
\end{bmatrix} 
+ \int_{-r}^{0} \bm{y}^\top(t+\tau) \la \bm{y}(t+\tau) \dif \tau 
\\
\leq \la \norm{\bm{x}(t)}_{2}^{2} + [*] \la \textstyle \int_{-r}^{0} F(\tau) \bm{y}(t+\tau) \dif \tau + \la \norm{\bm{y}(t+\dt)}_{\infty}^2   
\leq \la\norm{\bm{x}(t)}_{2}^{2} + \la \norm{\bm{y}(t+\dt)}_{\infty}^2 
\\
+ [*]\left( \eta \Ff \otimes I_n \right)  \textstyle \int_{-r}^{0} F(\tau) \bm{y}(t+\tau) \dif \tau
\leq \la\norm{\bm{x}(t)}_{2}^{2} + \la \norm{\bm{y}(t+\dt)}_{\infty}^2 
+   \textstyle \int_{-r}^{0} [*] \eta \bm{y}(t+\tau) \dif \tau 
\\
\leq \la\norm{\bm{x}(t)}_{2}^{2} + \left( \la + \eta r  \right)  \norm{\bm{y}(t+\dt)}_{\infty}^2
\leq \left( \la + \eta r  \right)\norm{\bm{x}(t)}_{2}^{2} + \left( \la + \eta r \right)  \norm{\bm{y}(t+\dt)}_{\infty}^2 
\\
\leq 2\left( \la + \eta r  \right) \left( \norm{\bm{x}(t)}_{2} \ma \norm{\bm{y}(t+\dt)}_{\infty} \right)^2.
\end{multline}
which demonstrates that \eqref{KL} satisfies  
\begin{equation}\label{app1}
\exists \ep_2 > 0: \forall t \in \mT, \;\; v(\bm{x}(t), \bm{y}(t+\dt))
\leq \ep_2 \left(  \bm{x}(t) \ma \norm{\bm{y}(t + \dt)}_{\mathit{\infty}} \right)^2.
\end{equation} 

Applying the Lemma 5 in \cite{Feng201660} to the integral term 
$\int_{-r}^{0} \bm{y}^\top(t+\tau) S \bm{y}(t+\tau) \dif \tau $ in \eqref{KL} with $U \succ 0 $ and the fact that $\bm{y}(t+\dt) \in \fC_{\circ}^{\bullet} ( [-r,0); \mathbb{R}^\nu) \subset \fset{L}_\mathit{2} ( [-r,0); \mathbb{R}^\nu)$, yields 
\begin{equation}\label{ini}
\textstyle \forall t \in \mT, \;\; \int_{-r}^{0} \bm{y}^\top(t+\tau) S \bm{y}(t+\tau) \dif \tau \geq \left( \int_{-r}^{0} F(\tau) \bm{y}(t+\tau) \dif \tau \right)^\top \Ff \otimes S \int_{-r}^{0} F(\tau) \bm{y}(t+\tau) \dif \tau.
\end{equation} By considering \eqref{ini} with \eqref{KL}, one can conclude that the feasible solution of \eqref{positiv} infers the existence of \eqref{KL} satisfies \eqref{posit} considering  the right limit substitution $t = \tau, \bm{x}(\tau) = \bm{\xi}, \bm{y}(\tau+\dt) = \bm{\fii}(\dt)$ and \eqref{app1}. 

Now we start to derive the stability conditions inferring \eqref{neg} and \eqref{diss}. 

Differentiate \( v(\bm{x}(t), \bm{y}(t+\dt)) \) alongside the trajectory of \eqref{openloop} and considering \eqref{supply} and the relation 
\begin{multline}\label{inte}
\textstyle \frac{\dif }{\dif t}\int_{t-r}^{t} F(\tau) \bm{y}(\tau) \dif \tau =  F(0)  \bm{y}(t) - F(-r) \bm{y}(t-r) 
-  (M \otimes I_{\nu})
\int_{-r}^{0} F(\tau) \bm{y}(t+\tau) \dif \tau 
= F(0) A_4\bm{x}(t) 
\\
\textstyle + \left[ F(0)A_5 - F(-r) \right] \bm{y}(t-r)
- \dhat{M} \int_{-r}^{0} F(\tau) \bm{y}(t+\tau) \dif \tau,
\end{multline}
where $\dhat{M} = M \otimes I_{\nu}$.  Then we have
\begin{gather}
\begin{aligned}\label{iquissen}
& \dot{v}(\bm{x}(t),\bm{y}(t+\cdot))  - s(\bm{z}(t),\bm{w}(t)) 
\\
& = \bm{\chi}^\top(t) \Sy \left( H P \Theta \right) \bm{\chi}(t) + \bm{\chi}^\top (t) \left[ \Ga^\top \!\! \left( S + rU \right) \Ga - \left( J_3 \oplus \Of_{n} \oplus S \oplus \Of_{\ro} \right) \right] \bm{\chi}(t)
\\
& - \bm{\chi}^\top(t) \left( \Sig^\top J_1 \Sig + \Sy\left( \begin{bmatrix}  \Sig^\top J_2 & \Of_{(n + \nu + \ro + q) \times (n + \nu + \ro) } \end{bmatrix} \right) \right] \bm{\chi}(t) - \textstyle \int_{-r}^{\,0} \bm{y}^\top(t+\tau) U  \bm{y}(t+\tau) \dif \tau,
\end{aligned}
\end{gather}	
where $ \Ga $ and $\Sigma$ have been defined in the statements of Theorem \ref{Theorem 1}
 and
\begin{equation}\label{chi}
\begin{gathered}
\bm{\chi}(t):= \Col 
	\textstyle \left( \bm{w}(t), \; \bm{x}(t), \; \bm{y}(t-r),  \;  \int_{-r}^{0} F(\tau) \bm{y}(t + \tau) \dif \tau  \right).
\end{gathered}
\end{equation}

Let $U \succ 0$ and apply the Lemma 5 in \cite{Feng201660} to the integral 
\( \int_{-r}^0 \bm{y}^\top(t+\tau)  U  \bm{y}(t+\tau) \dif \tau \) in \eqref{iquissen} similar to \eqref{ini}. It produces 
\begin{equation}\label{in2}
\textstyle \forall t \in \mT, \;\; \int_{-r}^{0} \bm{y}^\top(t+\tau) U \bm{y}(t+\tau) \dif \tau \geq \left( \int_{-r}^{0} F(\tau) \bm{y}(t+\tau) \dif \tau \right)^\top \left( \Ff \otimes U \right) \int_{-r}^{0} F(\tau) \bm{y}(t+\tau) \dif \tau.
\end{equation} 
Now considering \eqref{in2} with \eqref{iquissen}, we have
\begin{equation}\label{kendissen}
\forall t \in \mT,\;  \dot{v}(\bm{x}(t), \bm{y}(t+\cdot)) - s(\bm{z}(t),\bm{w}(t)) \leq 
\bm{\chi}^\top (t) \left(\bm{\Phi} -  \Sig^\top J_1 \Sig \right )
\bm{\chi}(t),
\end{equation}
where $\bm{\Fi}$ and $\Sig$ have been defined in \eqref{fi} and \( \bm{\chi} (t) \) have been defined in \eqref{chi}. Based on the structure of \eqref{kendissen}, it is easy to see that if  $U \succ 0$ and 
\begin{equation}\label{meitriks}
 \bm{\Phi} - \Sig^\top J_1 \Sig \prec 0, \;\; 
\end{equation}
are satisfied then the dissipative inequality in  
\eqref{diss} : \( \dot{v}(\bm{x}(t),\bm{y}(t+\cdot))  - s(\bm{z}(t),\bm{w}(t)) \leq 0 \) holds $\forall t \in \mT$. Furthermore, given $J_1 \prec 0$ with the structure of $\bm{\Phi} - \Sig^\top J_1 \Sig \prec 0$  and considering the properties of positive definite matrices, it is obvious that the feasible solution of  \eqref{meitriks} with $U \succ 0$ infers the existence of \eqref{KL} satisfying \eqref{diss} and \eqref{neg} considering the definition of \eqref{difer}.

On the other hand, given $J_1 \prec 0$, applying Schur complement to \eqref{meitriks} enables one to conclude that given $U \succ 0$, \eqref{meitriks} holds if and only if \eqref{meitrix}
which is now a convex matrix inequality. Since $U \succ 0$ is included in \eqref{positiv}, thus one can conclude that the feasible solutions of \eqref{positiv}--\eqref{meitrix} infer the existence of \eqref{KL} satisfying \eqref{posit},\eqref{neg} and \eqref{diss}.
\end{proof}

\subsection{Dissipative stability conditions with slack variables}
In this subsection, we derive dissipative conditions via the following Finsler and Projection Lemmas. The result can be considered as a  generalization of the approach in  \cite{de2001stability} to handle delay systems. Furthermore, we prove that the conditions with slack variables are in fact equivalent to  Theorem \ref{Theorem 1} in terms of feasibility.
 
\begin{lem}[Finsler Lemma \cite{de2001stability}] \label{Lemma 3} Given \( n;p;q \in \mathbb{N} \), $ \Pi \in \mathbb{S}^{n}$,  $P \in \mathbb{R}_{[q]}^{p \times n}$ such that $q < n$, then the following propositions are equivalent: 
\begin{gather}
\mathbf{x}^\top \Pi \mathbf{x} < 0, \forall \mathbf{x} \in \left\{ \mathbf{y} \in \mathbb{R}^{n} \!\setminus 
\! \{ \bm{0} \} : P \mathbf{y} = \bm{0}_m \right\} \label{ke}
\\[1mm]
\exists Y \in \mR^{n \times p}: \;\; \Pi +  \Sy( Y P) \prec 0, \label{mei}
\\
P_{\bot}^\top \Pi P_{\bot} \prec 0, \label{pra}
\end{gather}
where the columns of \( P_{\bot} \) contains any basis of the null space of \( P \).
\end{lem}

\begin{lem}[Projection Lemma] \cite{Feng201660} \label{Lemma 4} Given \( n;p;q \in \mathbb{N}, \; \Pi \in \mathbb{S}^{n},  P \in \mathbb{R}^{q \times n}, Q \in \mathbb{R}^{p \times n}\), there exists \(\Yp \in \mathbb{R}^{p \times q}\) such that the following two propositions are equivalent : 
\begin{align}
    & \Pi + P^\top \Yp^\top Q + Q^\top \Yp P \prec 0, \label{slack}
    \\[2mm]
    & P_{\bot}^\top \Pi P_{\bot} \prec 0 \;\; and \;\;  Q_{\bot}^\top \Pi  Q_{\bot} \prec 0, \label{projec}
\end{align}
\vspace{2mm}
where \( P_{\bot} \) and \( Q_{\bot} \) are matrices in which the columns contain any basis of the null space of \( P \) and \( Q \), respectively.
\end{lem}

\begin{thm}\label{Theorem 2}
Given all the prerequisites in Theorem \ref{Theorem 1}, \eqref{openloop} is globally uniformly asymptotically stable at its origin and dissipative with respect to \eqref{supply},  if there exist $Y \in \mR^{(2n + \nu + \ro + q) \times n}$ and $ P \in \mathbb{S}^{n + \ro}$ and $S;U \in \mathbb{S}^{\nu}$ such that \eqref{positiv} and the following inequality hold,
\begin{gather}\label{seff}
\begin{bmatrix} J_1^{-1} & \dtl{\Sig} \\ * & \dtl{\bm{\Fi}} \end{bmatrix} + \Sy \left( \begin{bmatrix} \Of_{(m + q) \times n} \\ Y \end{bmatrix} \begin{bmatrix} \Of_{n \times m} & \mathbf{A} & -I_n \end{bmatrix} \right) \prec 0, 
\end{gather}
where $\mathbf{A} = \begin{bmatrix} D_1  & A_1 & A_2 & A_3 \end{bmatrix}$ and $\dtl{\Ga} := \begin{bmatrix} \Ga & \Of_{\nu \times n} \end{bmatrix}$ and $\dtl{\Sigma} := \begin{bmatrix}
	\Sigma & \Of_{m \times n}	
	\end{bmatrix}$ with $\Ga$, $\Sig$ defined in Theorem \ref{Theorem 1}, 
and	 
\begin{multline}\label{til}
\bm{\dtl{\Fi}} = \Sy \left( \begin{bmatrix}
H \\ \Of_{ n \times (\ro + n)} \end{bmatrix} P 
\begin{bmatrix}
\Of_{n \times q}  & \Of_{n} & \Of_{n \times \nu} & \Of_{n \times \ro} & I_n
\\
\Of_{\ro \times q} & F(0)A_4 & F(0)A_5 - F(-r) & -M & \Of_{\ro \times n}
\end{bmatrix} \right) 
\\
+ \dtl{\Ga}^\top \!\! \left( S + rU \right) \dtl{\Ga} - \Big( J_3 \oplus \Of_{n} \oplus S \oplus (\Ff \otimes U) \oplus \Of_{n} \Big) - \Sy\left( \begin{bmatrix}  \dtl{\Sig}^\top J_2 & \Of_{(2n + \nu + \ro + q) \times (2n + \nu + \ro) } \end{bmatrix} \right) . 
\end{multline}
with $H$ and $\Ff$ defined in \eqref{F}.
\end{thm}

\begin{proof}

First of all, note that the conditions in \eqref{upp} and \eqref{positiv} remain unchanged in this case as they can be derived independently without considering the configuration of \eqref{openloop}.

Unlike what has been presented in the proof of Theorem \ref{Theorem 1}, 
$\dot{v}(\bm{x}(t), \bm{y}(t+\dt)) - s(\bm{z}(t), \bm{w}(t))$ can be also formulated as 
\begin{gather}
\begin{aligned}\label{iquisse}
& \dot{v}(\bm{x}(t),\bm{y}(t+\cdot))  - s(\bm{z}(t),\bm{w}(t)) 
\\
& = \bm{\eta}^\top(t) \Sy \left(  \begin{bmatrix}
H \\ \Of_{ n \times (\ro + n)} \end{bmatrix}  P 
\begin{bmatrix}
\Of_{n \times q}  & \Of_{n} & \Of_{n \times \nu} & \Of_{n \times \ro} & I_n
\\[1mm]
\Of_{\ro \times q} & F(0)A_4 & F(0)A_5 - F(-r) & -\dhat{M} & \Of_{\ro \times n}
\end{bmatrix} \right) \bm{\eta}(t) 
\\
& + \bm{\eta}^\top (t) \left[ \dtl{\Ga}^\top \!\! \left( S + rU \right) \dtl{\Ga} - \left( J_3 \oplus \Of_{n} \oplus S \oplus \Of_{\ro} \oplus \Of_{n} \right) \right] \bm{\eta}(t)
\\
& - \bm{\eta}^\top(t) \Big( \dtl{\Sig}^\top J_1 \dtl{\Sig} + \Sy\left( \begin{bmatrix}  \dtl{\Sig}^\top J_2 & \Of_{(2n + \nu + \ro + q) \times (2n + \nu + \ro) } \end{bmatrix} \right) \bm{\eta}(t), 
\\
& - \textstyle \int_{-r}^{\,0} \bm{y}^\top(t+\tau) U  \bm{y}(t+\tau) \dif \tau, \;\; \forall \bm{\eta}(t) \in \{ \bm{\eta}(\dt): \forall t \in \mT, \;\; \begin{bmatrix} \mathbf{A} & -I_n \end{bmatrix} \bm{\eta}(t) = 0 \},  \;\; \forall t \in \mT
\end{aligned}
\end{gather} 
where $\mathbf{A}$ and $\dtl{\Ga}$ and $\dtl{\Sigma}$ have been defined in the statement of Theorem \ref{Theorem 2} and $\bm{\eta}(t):= \Col \left( \bm{\chi}(t), \; \dot{\bm{x}}(t) \right)$ with $\bm{\chi}(t)$ defined in \eqref{chi}. Specifically, the information of \eqref{openloop} in \eqref{iquisse} is characterized by the dynamical constraints $\forall \bm{\eta}(t) \in \kE := \{ \bm{\eta}(\dt): \forall t \in \mT, \;\; \begin{bmatrix} \mathbf{A} & -I_n \end{bmatrix} \bm{\eta}(t) = 0 \}$ where $\dot{\bm{x}}(t)$ is part of $\bm{\eta}(t)$. 

Now applying \eqref{in2} to \eqref{iquisse} assuming $U \succ 0$ in \eqref{positiv} yields 
\begin{equation}\label{kenstrei}
\forall t \in \mT,\; \forall \bm{\eta}(t) \in \kE, \;   \dot{v}_1(\bm{x}(t), \bm{y}(t+\cdot)) - s(\bm{z}(t),\bm{w}(t)) \leq 
\bm{\eta}^\top (t) \left( \dtl{\bm{\Phi}} - \dtl{\Sig}^\top J_1 \dtl{\Sig} \right )
\bm{\eta}(t),
\end{equation}
where  $\dtl{\bm{\Fi}}$ is defined in \eqref{til}.
Furthermore, since \eqref{ke} and \eqref{pra} are equivalent, one can conclude that
\begin{gather} \label{suf}
 \bm{\thet}^\top \left[  \dtl{\bm{\Phi}} - \dtl{\Sigma}^\top J_1 \dtl{\Sigma}  \right] \bm{\theta} < 0, \;\; 
\forall \bm{\thet} \in \Big \{ \mathbf{y} \in \mathbb{R}^{2n + \ro + q} \!\setminus 
\! \{ \bm{0} \} :  \begin{bmatrix} \mathbf{A} & -I_n \end{bmatrix} \mathbf{y} = \bm{0}_{n} \Big\}
\end{gather}
holds if and only if 
\begin{gather}\label{suff}
\exists Y \in \mR^{(2n + \nu + \ro + q) \times n}:\;\;   \dtl{\bm{\Phi}} - \dtl{\Sigma}^\top J_1 \dtl{\Sigma}  
+ Y \begin{bmatrix} \mathbf{A} & -I_n \end{bmatrix} \prec 0,  
\end{gather}
where \eqref{suf} infers the existence of \eqref{KL} satisfying \eqref{neg} and \eqref{diss} considering the properties positive matrices and the right limit substitution $t = \tau, \bm{x}(\tau) = \bm{\xi}, \bm{y}(\tau+\dt) = \bm{\fii}(\dt)$.

Apply Schur complement to \eqref{suff} with $J_1 \prec 0$. It shows the inequality in \eqref{suff} is equivalent to
\begin{multline}\label{szhra}
\begin{bmatrix} J_1^{-1} & \dtl{\Sig} \\ * & \dtl{\bm{\Fi}} \end{bmatrix} + \Sy \left( \begin{bmatrix} \Of_{m \times n} \\ Y \end{bmatrix} \begin{bmatrix} \Of_{n \times m} & \mathbf{A} & -I_n \end{bmatrix} \right) 
\\
=  \begin{bmatrix} J_1^{-1} & \dtl{\Sig} \\ * & \dtl{\bm{\Fi}} \end{bmatrix} + \Sy \left( \begin{bmatrix} \Of_{m \times (2n + \nu + \rho + q)} \\ I_{2n + \nu + \rho + q} \end{bmatrix} Y \begin{bmatrix} \Of_{n \times m} & \mathbf{A} & -I_n \end{bmatrix} \right) \prec 0. 
\end{multline}
By the equivalence between \eqref{projec} and \eqref{slack}, we have \eqref{szhra} holds if and only if 
\begin{equation}\label{J1}
\begin{bmatrix} \Of_{m \times (2n + \nu + \rho + q)} \\ I_{2n + \nu + \rho + q} \end{bmatrix}_{\bot} \begin{bmatrix} J_1^{-1} & \dtl{\Sig} \\ * & \dtl{\bm{\Fi}} \end{bmatrix} [*] = \begin{bmatrix} I_m & \Of_{m \times (2n + \nu + \rho + q) }\end{bmatrix} \begin{bmatrix} J_1^{-1} & \dtl{\Sig} \\ * & \dtl{\bm{\Fi}} \end{bmatrix} [*] = J_1^{-1} \prec 0
\end{equation}
and 
\begin{equation}\label{eqpro}
[*] \begin{bmatrix} J_1^{-1} & \dtl{\Sig} \\ * & \dtl{\bm{\Fi}} \end{bmatrix} \begin{bmatrix} \Of_{n \times m} & \mathbf{A} & -I_n\end{bmatrix}_{\bot} = [*] \begin{bmatrix} J_1^{-1} & \dtl{\Sig} \\ * & \dtl{\bm{\Fi}} \end{bmatrix} \begin{bmatrix}  I_{m + q + n + \nu + \rho} \\ \begin{bmatrix} \Of_{n \times m} & \mathbf{A} \end{bmatrix} \end{bmatrix} \prec 0.
\end{equation}

Now realize that \eqref{J1} holds given \eqref{eqpro}. Thus \eqref{eqpro} is equivalent to \eqref{suff}.  On the other hand, note that \eqref{eqpro} gives 
\begin{multline}\label{iqpro}
\begin{bmatrix} \Of_{(m +q) \times (2n + \nu + \rho)} \\[1mm] I_{2n + \nu + \rho} \end{bmatrix}_{\bot} \begin{bmatrix} J_1^{-1} & \dtl{\Sig} \\ * & \dtl{\bm{\Fi}} \end{bmatrix} [*] = \begin{bmatrix} I_{m+q} 
& \Of_{(m +q) \times (2n + \nu + \rho) }\end{bmatrix} \begin{bmatrix} J_1^{-1} & \dtl{\Sig} \\ * & \dtl{\bm{\Fi}} \end{bmatrix} [*] 
\\
= \begin{bmatrix} J_1^{-1} & D_2 \\  * &  J_3 + \Sy ( D_2^{\top} J_2 ) \end{bmatrix} \prec 0 
\end{multline}
which also infers \eqref{J1}. By Lemma \ref{Lemma 4}, one can conclude that \eqref{eqpro} with \eqref{iqpro} are equivalent to 
\eqref{seff} which subsequently is equivalent to \eqref{szhra}. This in fact shows that there are redundant slack variables in \eqref{szhra} which can be reduced into the form of \eqref{seff}. Thus \eqref{suf} is equivalent to \eqref{seff} which finishes the proof. 
\end{proof}

\begin{coro}
 Theorem \ref{Theorem 1} is equivalent to Theorem \ref{Theorem 2} in terms of feasibility.
\end{coro}
\begin{proof}
Note that \eqref{meitrix} can be reformulated into 
\begin{equation}\label{iku}
\begin{bmatrix} J_1^{-1} & \Sig \\ * & \bm{\Fi}   \end{bmatrix} = [*] \begin{bmatrix} J_1^{-1} & \dtl{\Sig} \\ * & \dtl{\bm{\Fi}} \end{bmatrix} \begin{bmatrix}  I_{m + q + n + \nu + \rho} \\ \begin{bmatrix} \Of_{n \times m} & \mathbf{A} \end{bmatrix} \end{bmatrix}
\end{equation}
considering all the structures of \eqref{fi}, \eqref{til} with $\mathbf{A}$, $\Sig$ and $\dtl{\Sig}$. Specifically, by letting $P = \left[ \begin{smallmatrix} P_1 & Q \\ * & R \end{smallmatrix} \right]$ in \eqref{KL}, in which $P_1 \in \mS^{n}$, $Q \in \mS^{n \times \ro}$ and $R \in \mS^{\ro \times \ro}$,  \eqref{iku} can be derived similarly as the equations (35)--(37) in \cite{Feng201660}. Thus \eqref{eqpro} is equal to \eqref{meitrix}. Since \eqref{eqpro} is equivalent to \eqref{seff}, this finishes the proof.
\end{proof}

\section{Conclusion}
Dissipative conditions for a linear CDDS have been derived based on a dynamical constraints approach. The result can be considered as a generalization of the Finsler Lemma approach in \cite{de2001stability} towards delay related system. Moreover, we have shown that the dissipative conditions with slack variables are equivalent to the conditions derived by direly substituting $\dot{\bm{x}}(t)$ during the construction of the  Krasovskii functional.

\bibliographystyle{authordate1}

\end{document}